\documentclass[envcountsame]{llncs}
\usepackage{amsmath,amssymb}
\usepackage[lined,boxed,linesnumbered]{algorithm2e}
\usepackage{multirow}
\usepackage[utf8]{inputenc}

\begin{document}
\spnewtheorem{fact}[theorem]{Fact}{\bfseries}{\itshape}
\spnewtheorem{observation}[theorem]{Observation}{\bfseries}{\itshape}

\newcommand{\ie}{\textit{i.e.}}
\newcommand{\eg}{\textit{e.g.}}

\title{Approximation algorithms for node-weighted prize-collecting Steiner tree problems on planar graphs}

\author{Jarosław Byrka\inst{1}, Mateusz Lewandowski\inst{1}, Carsten Moldenhauer\inst{2}}

\institute{University of Wrocław, Poland
\and
EPFL, Lausanne, Switzerland}

  \pagestyle{plain}

\maketitle
\vspace{-6mm}
\begin{abstract}
	We study the prize-collecting version of the Node-weighted Steiner Tree problem (NWPCST) restricted to planar graphs.
	We give a new primal-dual Lagrangian-multiplier-preserving (LMP) 3-approximation algorithm for planar NWPCST.
	We then show a ($2.88 + \epsilon$)-approximation which establishes a new best approximation guarantee for planar NWPCST.
	This is done by combining our LMP algorithm with a threshold rounding technique and utilizing the
	2.4-approximation of Berman and Yaroslavtsev~\cite{FVSplanar2.4} for the version without penalties.
	We also give a primal-dual 4-approximation algorithm for the more general forest version
	using techniques introduced by Hajiaghay and Jain~\cite{Hajiaghayi}.
\end{abstract}
\vspace{-8mm}
\section{Introduction}
\vspace{-4mm}
	In Steiner problems we aim at connecting certain specified vertices (called terminals) by buying edges or nodes of the given graph. The classic edge-weighted setting is well known to have many applications in areas like electronic circuits, computer networking, and telecommunication. The expressive power of the node weighted variants is used to model various settings common to bioinformatics \cite{bioinformatics}, maintenance of electric power networks \cite{powerOutage}, and computational sustainability \cite{sustainability}.

The node weighted setting is a generalization of the edge weighted case. In particular, one may cast the Set Cover problem as an instance of the Node-weighted Steiner Tree problem, which proves hardness of approximation of the general node-weighted setting. In this paper we study a natural special case, namely planar graphs, for which constant factor approximation algorithms are possible. 

In the prize-collecting (penalty-avoiding) setting we are given an option not to satisfy a certain connectivity requirement, but to pay a fixed penalty instead. The main focus of this work is to develop efficient primal-dual approximation algorithms for prize-collecting versions of the node-weighted Steiner problems. 

\subsection{Previous work}
The Steiner Tree problem is NP-hard even in planar graphs \cite{hardnessSTevenPlanar}. 
The most studied is the standard Edge-weighted Steiner Tree, for which the best known approximation ratio 1.39 is obtained via a randomized iterative rounding technique~\cite{Byrka}. By contrast, the best approximation algorithms for Steiner Forest have the so far unbreakable ratio of 2~\cite{SteinerForest,JainIterative}.

For the Prize-collecting Steiner Tree problem there exists a primal-dual 2-approximation algorithm \cite{Goemans}. It can be shown that it is also Lagrangian-preserving. This property was used by Archer et al. to design the currently best $2-\epsilon$ approximation algorithm for PCST \cite{PCST}.

For the Prize-collecting Steiner Forest problem there is a $3$ approximation primal-dual algorithm \cite{Hajiaghayi}, which introduces a general technique to handle prize-collecting problems. In the same paper the authors use a threshold rounding technique with randomized analysis to obtain $\approx 2.54$ approximation.
		
There are optimal (up to a constant factor) algorithms for node-weighted Steiner problems. One example is the recent $O(\ln n)$ approximation algorithm for NWPCSF by Bateni et al~\cite{Bateni}.
K\"onemann et al~\cite{Konemann} gave a Langrangian-multiplier-preserving (LMP) approximation that achieves the same guarantee.
Establishing the LMP property is of crucial importance for the construction of approximation algorithms for quota and budgeted versions
of the NWST problem.

Planarity helps significantly in both edge and node weighted setting. Both ST and SF admit PTAS in planar graphs \cite{SFptas}. Planar PCST can be also approximated with any constant, but PCSF is APX-HARD already in planar graphs~\cite{PCSTptasFhardness}. 

Planarity allows for constant factor approximations for node-weighted Steiner problems. The $NWSF$ can be expressed as the Hitting Set problem for some uncrossing family of cycles and hence solved as a feedback problem. This was exploited by Berman and Yaroslavtsev in \cite{FVSplanar2.4} where they obtained $2.4$ approximation for $NWSF$ and other problems on planar graphs.

In~\cite{MoldenhauerThesis} it was observed that using a threshold rounding technique together with
the 2.4-approximation of Berman and Yaroslavtsev~\cite{FVSplanar2.4} for the version without penalties
gives a 2.93-approximation algorithm for NWPCST on planar graphs. This was the best approximation guarantee up to date.
However, such an algorithm requires solving an LP.

We summarize the current best known results in Table~\ref{table:literature}.

\begin{table}[]
\centering
\label{table-results}
\begin{tabular}{cc|c|c|c|c|}
\cline{3-6}
                                                                                       &                                       & \multicolumn{2}{c|}{{\bf Edge-weighted}}                                         & 
\multicolumn{2}{c|}{{\bf Node-weighted}}                                              \\ \cline{3-6} 
                                                                                       &                                       & {\bf Tree} & {\bf Forest}                                & {\bf Tree}                    & {\bf Forest}                  \\ \hline
\multicolumn{1}{|c|}{}                                         & {\bf General} & 1.39 \cite{Byrka}                               & 2 \cite{Goemans}                                                                  & O(log k) \cite{Bateni}                                             & O(log k) \cite{Bateni}                                             \\ \cline{2-6} 
\multicolumn{1}{|c|}{\multirow{-2}{*}{\bf }}                 & {\bf Planar}  & PTAS \cite{SFptas}                               & PTAS \cite{SFptas}                                                               & 2.4 \cite{FVSplanar2.4}                                & 2.4 \cite{FVSplanar2.4}                                                  \\ \hline
\multicolumn{1}{|c|}{}                                         & {\bf General} & $2-\epsilon $ \cite{PCST}                       & \begin{tabular}[c]{@{}c@{}}3 \\ 2.54 (LP) \cite{Hajiaghayi}\end{tabular} & O(log k) \cite{Bateni,Konemann}                                              & O(log k) \cite{Bateni}                                             \\ \cline{2-6} 
\multicolumn{1}{|c|}{\multirow{-2}{*}{{\bf Prize-collecting}}} & {\bf Planar}  & PTAS \cite{PCSTptasFhardness}                              & APX-HARD \cite{PCSTptasFhardness}                                                           & \begin{tabular}[c]{@{}c@{}}\textbf{3}\\ \textbf{2.87+$\epsilon$ (LP)}\end{tabular} & \begin{tabular}[c]{@{}c@{}}\textbf{4}\end{tabular} \\ \hline
\smallskip
\end{tabular}
\caption{Summary of best known approximation ratios for Steiner problems. Results of this paper are highlighted.}
\label{table:literature}
\end{table}

\subsection{Our contribution}
We propose a new LMP $3$-approximation algorithm for NWPCST on planar graphs. The algorithm is an adaptation of the original technique developed by Goemans and Williamson in \cite{Goemans} for PCST to the node-weighted version.
However, we change the pruning phase of the algorithm. This enables us to analyze the connection and penalty
costs separately which is the key ingredient.
In particular, we can directly charge the penalty costs to a part of the dual solution yielding Langrangian-multiplier-preservation.
Further, the connection costs can be bounded using a slightly adapted analysis from~\cite{Moldenhauer} for NWSF.
The approximation ratio of~$3$ is slightly higher than the previously best approximation ratio
but the primal-dual algorithm does not require solving an LP.

Next, we establish a new best approximation ratio by exploiting the asymmetry of our primal-dual algorithm.
Binding two different linear programs together permits a careful combination of the new LMP algorithm with a threshold
rounding technique. Finally, exploiting the 2.4-approximation from~\cite{FVSplanar2.4} we obtain a
(2.88 + $\epsilon$)-approximation for NWPCST on planar graphs.

Furthermore, we obtain an efficient, direct primal-dual $4$-approximation algorithm for NWPCSF on planar graphs building up on ideas for edge-weighted PCSF from~\cite{Hajiaghayi}. We defer the details of this result to Appendix~\ref{App:AppendixForest}.
This approach was previously indicated by Demaine et al.~\cite{Demaine}, but we give a better constant.

\section{The LMP primal-dual 3-approximation algorithm}
\vspace{-4mm}
	\label{section:nwpcst}
	Consider an undirected graph $G=(V,E)$ with non-negative cost function and penalties on the vertices denoted by $w : V \rightarrow Q_+$ and $\pi : V \rightarrow Q_+$, respectively.
	In the NWPCST problem we are allowed to purchase a \emph{connected} subgraph $F$ of $G$ that connects vertices to a prespecified root $r\in V$. Every bought vertex induces a cost according to $w$. Every vertex that is not included induces a penalty according to $\pi$. The objective is to minimize the sum of the purchase and penalty costs, \ie,
$\sum_{v\in F} w_v + \sum_{v \notin F}\pi_v$.

	By a standard transformation we can assume that for every vertex $v$ either its cost or its penalty is zero. To see this consider a single vertex $v$ with both strictly positive cost and penalty. Add an additional vertex $v'$, set its cost to zero and penalty to $\pi_v$, add an edge from $v'$ to $v$ and set the penalty of $v$ to zero. Now, any solution in the original graph can be transformed to a solution of the same cost in the modified graph and vice-versa.

	In the sequel, we call a vertex with a positive penalty a \emph{terminal}. Terminals and the root can be purchased for free. Other vertices do not have a penalty and we call them non-terminals or Steiner vertices.

	Let $\Gamma(S)$ denote the set of neighbors of $S$, \ie, the set of vertices in $V\setminus S$ incident to vertices from $S \subseteq V$. Let also $\Pi(X) = \sum_{v \in X}\pi_v$. Thus, NWPCST is the following problem:
	\begin{align}
	\text{min} & \sum_{v\in V} w_v x_v + \sum_{X \subseteq V\setminus\{r\}} \Pi(X)z_X & (IP_{PCST}) \nonumber\\
	s.t.\nonumber\\
		& \sum_{v\in \Gamma(S)} x_v + \sum_{X: S \subseteq X} z_X \geq  1			& \forall S \subseteq V\setminus\{r\} \nonumber\\
		& x_v \in \{0,1\}															& \forall v \in V \nonumber\\
		& z_X \in \{0,1\}															& \forall X \subseteq V\setminus\{r\} \nonumber
	\end{align}

	By relaxing the integrality constraints to non-negativity constraints we obtain the standard linear relaxation.
	The dual of this relaxation is
	\begin{align}
	\text{max} & \sum_{S \subseteq V\setminus\{r\}} y_{S} & (DLP_{PCST}) \nonumber\\
	s.t. \nonumber\\
		& \sum_{S:v\in\Gamma(S)} y_S \leq w_v		& \forall v \in V \label{Constraint1}\\
		& \sum_{S\subseteq X} y_S \leq \Pi(X)		& \forall X \subseteq V\setminus\{r\} \label{Constraint2}\\
		& y_S \geq 0								& \forall S \subseteq V\setminus\{r\} \nonumber
	\end{align}

	\subsection{Algorithm}
		Now we shortly describe our primal-dual algorithm which is an adaptation of the generic moat-growing approach of Goemans and Williamson \cite{Goemans}. In each iteration $i$ we maintain a set of already bought nodes $F$. We say that some vertex was bought at time $i$ if it was bought in iteration $i$
		\footnote{When we refer to time we always have in mind the number of the current iteration. Note that it implies that the speed of the uniform growth of dual budgets is not constant across iterations,
		but it does not affect our description of the algorithm.}. At the beginning $F$ contains all terminals (including root). We maintain also the set of connected components $C$ of subgraph $G[F]$ induced by the vertices bought so far. We call each of this connected components a moat. Moats can be active or inactive. The moat containing root $r$ is always inactive. In each iteration we increase (grow) dual variables corresponding to \emph{all active} moats uniformly until one of the following two events happen:
		\begin{itemize}
			\item a vertex $v$ goes tight (constraint (\ref{Constraint1}) becomes equality), or
			\item a set $X$ goes tight (constraint (\ref{Constraint2}) becomes equality).
		\end{itemize}
		In the first case we buy vertex $v$ and possibly merge moats incident to $v$. If we merge to a moat containing the root $r$, this moat becomes inactive, otherwise it is declared active.

		In the second event we make the moat corresponding to set $X$ inactive. Moreover, we mark all unmarked terminals inside $X$ with the current time.

		The growth phase terminates when there are no more active moats. After that, we have a pruning phase. In the pruning phase we let $F^{(r)}$ be the connected component of $F$ containing the root. Then, we consider vertices in $F^{(r)}$ in the reverse order of purchase. We delete vertex $v$ (bought at time $t$) if it does not disconnect from $r$ any terminal which was unmarked at time $t$. When we delete~$v$, we delete also all vertices that become disconnected from $r$. As a result we output the set of bought vertices $F'$ that survived pruning.\\

		Our algorithm can be implemented with a notion of so-called potentials. Let $P(X) = \Pi(X) - \sum_{S \subseteq X} y_S$ be the potential of set $X$. Intuitively, we pay for the growth of moats (increase of dual variables) with potentials of these moats. If the potential of a moat goes to zero, the corresponding constraint becomes tight, so we have to make this moat inactive. When we merge moats to a new moat $S$ by buying a vertex, we compute the potential of $S$ as the sum of potentials of old moats.
	\subsection{Analysis}
		\begin{theorem} (Lagrangian multiplier preservation)
		\label{thm3aprox}
		Let $G$ be planar. The algorithm described in the previous section outputs a set of vertices $F'$ such that
		$$\sum_{v\in F'}w_v + 3 \Pi(V\setminus F') \leq 3\sum_{S\subseteq V\setminus \{r\}}y_S \leq 3\ OPT$$
		\end{theorem}

In the proof we want to use the obtained dual solution $y$ to account for the connection costs
and penalties of the primal solution $F'$. We will partition the $y_S$ into two sets. The first set will yield a bound
on the connection costs and the second a bound on the penalties.

The key ingredient in the analysis is the partition that is based on the following lemma.
Consider any iteration $i$ and the active moats $A_i$ before this iteration.
Let $S\in A_i$ be an active moat that was not included in the final solution, \ie, $S\cap F'=\emptyset$.
Then, the dual variable of $S$ did not contribute to buying any vertex in $F'$. This means that $y_S$ does not
contribute to the left-hand-side of the constraints \eqref{Constraint1} for any $v\in F'$.
More formally, this means that $S$ does not have a neighbor in $F'$.

\begin{lemma}
	\label{lemma-neighbors}
	Let $S\in{A_i}$ be such that $S\subseteq V \setminus F'$. Then, the moat $S$ does not have
	any neighbor in the solution, i.e. $F' \cap \Gamma(S) = \emptyset$
\end{lemma}
\begin{proof}[of Lemma~\ref{lemma-neighbors}]\\
	Note that $S\in A_i$ means that $S$ is active in iteration $i$ and therefore there is an unmarked (before
	time $i$) terminal in $S$.
	Now, assume for a contradiction that $F'\cap \Gamma(S) \neq \emptyset$ and let $U\subseteq S$ be the set of
	vertices having a neighbor in $F'$.
	Note that all vertices in $U$ were bought before iteration $i$ because $S$ is a connected component of the
	vertices bought before iteration $i$ and $U\subseteq S$.
	Since $S$ is not part of $F'$, all the vertices in $U$ must have been deleted in the pruning phase.
	A contradiction, since this would disconnect the unmarked (before time $i$) terminal in $S$. \qed
\end{proof}

Following Lemma~\ref{lemma-neighbors}, we can partition all dual variables into the variables that contributed to
buying the vertices of $F'$ and the dual variables that account for the penalties induced by $F'$.
Let $CC$ be the set of all moats $S\subseteq V\setminus\{r\}$ that include a vertex of $F'$ or have a neighbor in $F'$, \ie, $(S\cup \Gamma(S))\cap F'\neq\emptyset$ and $y_S>0$.
Let $PC$ be the set of all other moats, \ie, sets $S$ with $y_S>0$ but $S\not\in CC$. We will show that
\begin{align*}
	\sum_{v\in F'} w_v \le 3\ \sum_{S\in CC} y_S \qquad \text{and} \qquad \Pi(V\setminus F') = \sum_{S\in PC} y_S
\end{align*}
which yields Theorem~\ref{thm3aprox}.

To show the bound on the connection cost we perform the following thought experiment.
Consider the subgraph $G'$ of $G$ obtained by restricting to vertices from $V'=V\setminus (\cup_{S\in PC} S)$, \ie,
restricted to only the root and vertices in the moats in $CC$ that contribute to the connection costs.
Lemma~\ref{lemma-neighbors} implies that there is no edges between moats in $PC$ and $V'$.
Recall that in each iteration, the algorithm increases \emph{all} active moats. Hence, the run of the algorithm
restricted to $G'$ is exactly the same as running the algorithm directly on $G'$.
Formally, let $(H,y')$ be the primal and dual solution obtained by running the algorithm on $G'$.
Then, $H = F'\cap V'$ and $y' = y|_{S\subseteq V'}$.

Now, we can leverage the analysis of the primal-dual algorithm for Node-weighted Steiner Forest given in~\cite{Moldenhauer}.
Recall that a terminal is a vertex with strictly positive penalty. Let $T$ be the set of terminals that are in any
moat of $CC$. Note that all vertices in $T$ are connected to the root since the moats in $CC$ were not disconnected
in the pruning phase. However, the \emph{execution} of our algorithm on $G'$ is \emph{not} the same as running the
primal-dual algorithm for Steiner Forest on $G'$ with terminal pairs $(r,t)$ ($t\in T$).
This is because our algorithm is allowed to deactivate moats due to the
penalty constraints. But, the \emph{analysis} of an iteration of both algorithms is essentially analog. Intuitively,
deactivating a moat compares to satisfying a demand pair in the Forest problem. The proof of the following lemma only
requires a minor change to the analysis and we therefore defer it to Appendix~\ref{App:AppendixB}.

Note that the crucial point is that we increase the dual variables of \emph{all} active moats. This guarantees that
the algorithm run on input subgraph $G'$ is the same as the run on input $G$ with restricted view on $G'$.
Choosing just a subset of the active moats can break this property since in each iteration we do not know in advance which
moats will be pruned during the pruning phase. Therefore, it is not straight forward to include the advanced
violation oracles from~\cite{FVSplanar2.4} that select only a subset of the active moats for increase.

\begin{lemma}[analog of analysis in~\cite{Moldenhauer}]
\label{lemMainCC}
	Let $F'$ be the output of the algorithm and $A_i$ be the set of active moats before running iteration $i$.
	Then,
	\begin{align*}
		\sum_{S\in A_i \cap CC} | F' \cap \Gamma(S) | \leq 3 | A_i \cap CC |.
	\end{align*}
\end{lemma}

To conclude the upper bound on the connection costs, note that constraint~\eqref{Constraint1}
is tight for all vertices $v\in F'$. This gives
\begin{align*}
	\sum_{v\in F'}w_v & = \sum_{v\in F'} \sum_{S:v\in\Gamma(S)}y_S
	= \sum_{S\subseteq V\setminus \{r\}}|F' \cap \Gamma(S)|\ y_S
	= \sum_{S\in CC} |F'\cap \Gamma(S)|\ y_S.
\end{align*}
We will show that $\sum_{S\in CC} |F'\cap \Gamma(S)|\ y_S \le 3\ \sum_{S\in CC} y_S$ by induction on the number
of iterations. At the beginning all dual variables are equal to $0$ and the inequality holds. In iteration $i$ we grow each active moat from $A_i\cap CC$ by $\epsilon_i$. This increases the left-hand side by
$\epsilon_i \sum_{S\in A_i\cap CC}|F' \cap \Gamma(S)|$ and the right-hand side by $3 \epsilon_i|A_i\cap CC|$.
Then, Lemma~\ref{lemMainCC} concludes the proof of the bound on the connection costs.

In order to prove the bound on the penalties we employ the following lemma.
\begin{lemma}
	Let $F'$ and $y_S$ be the primal and dual solution constructed by the algorithm.
	The set of vertices $X = V \setminus F'$ not spanned by the final solution can be partitioned into sets $X_1, X_2, \dots X_l$ such that the potential of each set is~$0$, i.e., $P(X_k) = 0$ for each $k$.
\end{lemma}
\begin{proof}
	Observe that there are two ways for a vertex $v$ to be in $X$: either it was never a part of the root component ($v \in V\setminus F^{(r)}$) or it was deleted in the pruning phase ($v \in F^{(r)}$).
	It is easy to see that $P(V\setminus F^{(r)}) = 0$. Each vertex in $V\setminus F^{(r)}$ was at the end a part of some inactive component not containing the root and hence the potentials of these components were $0$. Or, it was never in any moat.

	It remains to show that the set $S$ of vertices disconnected from $F'$ by pruning a vertex $v$ can be partitioned into sets $X_k$ for which $P(X_k)=0$.
	Let $t$ be the time when $v$ was bought. Observe that every vertex $u$ in the neighborhood $\Gamma(S)$ of $S$ has
	been bought after time $t$ or was not bought at all. Now, $S$ contains only marked terminals at time $t$,
	otherwise $v$ would not have been pruned. Hence, $S$ is a union of inactive moats at time $t$.
	This gives the desired partition. \qed
\end{proof}

Observe that the sets $X_i$ are disjoint from $CC$ and that $PC$ is the set of all $S\subseteq X_i$ with $y_S>0$.
To conclude the bound on the penalties note that since all $X_k$ have zero potential we have
\begin{align*}
	\Pi(V\setminus F') = \sum_{k=1}^l \Pi(X_k) = \sum_{k=1}^l \sum_{S\subseteq X_k} y_S = \sum_{S\in PC} y_S.
\end{align*}

\section{Combination with threshold rounding}
\vspace{-4mm}
	A standard technique to generalize primal-dual algorithms from Steiner Tree problems to their price-collecting variations
is to use \emph{threshold rounding} (see Section 5.7 of~\cite{WilliamsonBook} or~\cite{marketChoice}).
Here, in a first step an LP formulation for the price-collecting version is solved over fractional variables. Then, we pick
a threshold $\alpha$ and consider the vertices that are bought with value at least $\alpha$ to be terminals.
In a second step, the primal-dual algorithm for the original Steiner Tree problem is run on this set of terminals to
obtain the final solution. We note that the resulting algorithm is deterministic because we can try all possible thresholds
(at most one for every vertex). However, the analysis uses a randomization argument.

We observed in~\cite{MoldenhauerThesis} that using threshold rounding in combination with the primal-dual $2.4$-approximation
for Node-weighted Steiner Forest by Berman and Yaroslavtsev~\cite{FVSplanar2.4} yields a $2.93$-approximation for NWPCST on planar graphs.

In this section, we combine the previous LMP algorithm with the threshold rounding technique to gain an improved approximation factor
of~$2.88$. Our approach is inspired by an idea of Goemans~\cite{Combining}. Intuitively, such an improvement is possible because
the LMP approximation improves over the factor of $3$ if the optimal solution induces a high penalty cost.
In contrast, if the penalties are only a small part of the optimal solution's cost, threshold rounding
can leverage the robustness of the underlying $2.4$-approximation. Thus, by combining the two algorithms we can hedge
their weaknesses.

However, there is a technical difficulty. Applying threshold rounding to the LP that was used for the analysis of the LMP 3-approximation
($LP_a$ below) is not straight forward. We circumvent this problem by considering a stronger LP (see $LP_b$ below) that is suitable
for threshold rounding. To link the two different formulations we will guess the cost of the optimal solution to $LP_a$ and restrict
$LP_b$ to have a similar objective value. More precisely, we will solve multiple versions of $LP_b$ (see $LP_b^k$ below) and
then apply threshold rounding to gain a solution.
To obtain the final solution, we simply take the best of all solutions stemming from $LP_b^k$ and the LMP 3-approximation.
We remark that the resulting algorithm is deterministic. However, for the analysis, we will use a randomized argument to combine
the bounds of all solutions and gain an approximation factor of $2.88$.

\subsection{Two Linear Programs}
\vspace{-3mm}
	Consider the LP used in the construction of the primal-dual LMP 3-approximation which we denote by $LP_a$.
	\begin{align*}
	\text{min} & \sum_{v\in V} w_v x_v + \sum_{X \subseteq V\setminus\{r\}} \Pi(X)z_X & (LP_a)\\
	s.t.\\
		& \sum_{v\in \Gamma(S)} x_v + \sum_{X: S \subseteq X} z_X \geq  1			& \forall S \subseteq V\setminus\{r\} \\
		& x_v \geq 0 \hspace{3mm} \forall v \in V \hspace{3mm}	& z_X \geq 0 \hspace{3mm} \forall X \subseteq V\setminus\{r\}
	\end{align*}
	Let further $LP_b$ be the following LP that lends itself to threshold rounding
	\begin{align*}
	\text{min} & \sum_{v\in V} w_v x_v + \sum_{u \in V\setminus\{r\}}\pi_{u}y_u & (LP_b)\\
	s.t.\\
		& \sum_{v\in \Gamma(S)} x_v + y_u \geq  1				& \forall S \subseteq V\setminus\{r\}, \hspace{3mm} u \in S \\
		& x_v \geq 0 \hspace{3mm} \forall v \in V	& y_u \geq 0 \hspace{3mm} \forall u \in V
	\end{align*}
	While we do not know how to solve $LP_a$ we can solve $LP_b$ to optimality using, \eg, the ellipsoid method.
	We remark that the algorithm which will be described in the sequel only requires to solve multiple instances
	of a variation of $LP_b$. $LP_a$ is solely used in the analysis to combine the threshold rounding with the LMP 3-approximation.
	\begin{fact}
	\label{strongerLP}
	$LP_a$ is stronger than $LP_b$, \ie, every feasible solution to $LP_a$ is also feasible to $LP_b$.
	\end{fact}
	\begin{proof}
	Let $(x,z)$ be feasible to $LP_a$. Set $y_u = \sum_{X: u\in X} z_X$. We claim that $(x,y)$ is feasible to $LP_b$.
	Consider any $S\subseteq V\setminus\{r\}$ and $u\in S$. We have
	\begin{align*}
	y_u = \sum_{X: u\in X} z_X \geq \sum_{X: S \subseteq X} z_X
	\end{align*}
	Moreover, the objective values of $(x,z)$ and $(x,y)$ in their respective formulations are equal
	\begin{align*}
	\sum_{u \in V}\pi_{u}y_u = \sum_{u \in V}\pi_{u} \sum_{X: u\in X} z_X = \sum_{X \subseteq V\setminus\{r\}} \sum_{u \in X}\pi_u z_X = \sum_{X \subseteq V\setminus\{r\}} \Pi(X)z_X.
	\end{align*} \qed
	\end{proof}

\subsection{Threshold rounding}
\vspace{-3mm}
We will first describe how to link the two different LP formulations and then apply threshold rounding.
In the sequel, let $(x^*, z^*)$ be the optimum solution to $LP_a$ with objective value $OPT_a$.
Further, if $T$ is a solution to NWPCST, let $w(T)$ be the total connection and $\pi(V\setminus T)$
be the total penalties of $T$. We also use this notation for (fractional) solutions: $w(x)$, $\pi(z)$ and $\pi(y)$.
\subsubsection{Binding the two LPs.}
\vspace{-3mm}
For comparison with the LMP 3-approximation we require a bound on the threshold rounding solution with respect to
$(x^*,z^*)$, the optimal solution to $LP_a$, which requires to link $LP_a$ and $LP_b$. This is done by
guessing the value of $w(x^*)$ and restricting $LP_b$ to find a solution $(x^\#,y^\#)$ with objective function
value close to $w(x^*)$. For given $k$ consider $LP_b^k$ defined as
\begin{align*}
\text{min} & \sum_{v\in V} w_v x_v + \sum_{u \in V\setminus\{r\}}\pi_{u}y_u & (LP_b^k)\\
s.t.\\
	& \sum_{v\in \Gamma(S)} x_v + y_u \geq  1						& \forall S \subseteq V\setminus\{r\}, \hspace{3mm} u \in S \\
	& \sum_{v\in \Gamma(S)} w_v x_v \in \left[(1+\epsilon)^k, (1+\epsilon)^{k+1}\right)	& \forall S \subseteq V\setminus\{r\} \\
	& x_v \geq 0 \hspace{3mm} \forall v \in V & y_u \geq 0 \hspace{3mm} \forall u \in V
\end{align*}
Note that the number of different $k$ which we need to consider is bounded by a polynomial in the size of the input.
Therefore, assume that $k$ is set such that $w(x^*) \in \left[(1+\epsilon)^k, (1+\epsilon)^{k+1}\right)$.
Let further $(x^\#, y^\#)$ be the optimal solution to $LP_b^{k+1}$.
\begin{fact}
	$w(x^\#) \leq (1+\epsilon)^2 w(x^*)$.
\end{fact}
\begin{fact}
	$\pi(y^\#) \leq \pi(z^*)$.
\end{fact}
\begin{proof}
	Consider the feasible solution $s = (x^*,y^*)$ to $LP_b$ which is derived from $(x^*, z^*)$ using the construction from Fact~\ref{strongerLP}. Due to this construction we have that $\pi(y^*) = \pi(z^*)$. Let $c$ be such that $c \cdot w(x^*) = w(x^\#)$. Now consider $s' = (c \cdot x^*,y^*)$ which is feasible to $LP_b$ since $c>1$. Moreover, $s'$ is also feasible to $LP_b^{k+1}$. Thus $w(x^\#)+\pi(y^\#) \leq w(c \cdot x^*)+\pi(y^*) = w(x^\#)+\pi(y^*)$. This concludes the proof. \qed
	\end{proof}

\subsubsection{Threshold rounding.}
\vspace{-3mm}
We use the standard threshold rounding technique~(cf. \cite{WilliamsonBook}). Let $\beta\in (0,1)$ be a constant to be determined later.
For every possible value $\alpha$ of $y^\#$ that is at most $\beta$, let
$Q = \{ u : y^\#_u \leq \alpha \}$. Consider the instance $I_{NWST_Q}$ of the $NWST$ problem which is derived
from $I_{NWPCST}$ by keeping only terminals from~$Q$. Let $LP_{NWST_Q}$ be the following linear program
\begin{align*}
\text{min} & \sum_{v\in V} w_v x_v & (LP_{NWST_Q})\\
	& \sum_{v\in \Gamma(S)} x_v \geq  1			& \forall S \subseteq V\setminus\{r\}, \hspace{3mm} Q \cap S \neq \emptyset \\
	& x_v \geq 0 								& \forall v \in V
\end{align*}
Let $OPT_{LP_Q}$ be the optimum objective function value of $LP_{NWST_Q}$.
We run the $2.4$-approximation algorithm for $I_{NWST_Q}$ by Berman and Yaroslavtsev~\cite{FVSplanar2.4} which
returns a solution $F$ such that its cost is no greater than $2.4 \cdot OPT_{LP_Q}$.
Finally, return the best of all obtained solutions $F$ (due to different values of $\alpha$).

Though the algorithm is deterministic its analysis is based on a randomized argument.
Instead of trying all possible values of $\alpha$, consider $\alpha$ to be chosen uniformly at random from $[0, \beta]$.
Consider $x' = \frac{1}{1-\alpha} x^\#$. It follows that $x'$ is a feasible solution to $LP_{NWST_Q}$.
We bound the expected connection and penalty costs of $F$.
\begin{align*}
	\mathbb{E}\left[\sum_{v \in F} w_v\right]
	&\leq
	\mathbb{E}\left[2.4 \cdot OPT_{LP_Q}\right] 
	\leq
	\mathbb{E}\left[2.4 \sum_{v \in V} x_v' \cdot w_v\right]
	\leq
	\mathbb{E}\left[\frac{2.4}{1-\alpha}\right] \sum_{v \in V}x^\#_v \cdot w_v \\
	&=
	\left( \int_0^{\beta} \frac{1}{\beta} \cdot \frac{2.4}{1-\alpha} d\alpha \right) w(x^\#) = \frac{2.4}{\beta} \ln\left(\frac{1}{1-\beta}\right) w(x^\#)\\
	&\leq
	\frac{2.4}{\beta} \ln\left(\frac{1}{1-\beta}\right) (1+\epsilon)^2 w(x^*)
\end{align*}
\begin{align*}
	\mathbb{E}\left[\sum_{u \notin Q} \pi_u \right]
	& =
	\mathbb{E}\left[\sum_{u : y^\#_u > \alpha} \pi_u \right]
	\leq
	\sum_u \pi_u Pr\left[ y^\#_u \geq \alpha \right]
	\leq
	\sum_u \pi_u \int_0^{y^\#_u} \frac{1}{\beta} d \alpha \\
	& =
	\sum_u \pi_u \frac{1}{\beta} y^\#_u = \frac{1}{\beta} \pi(y^\#) \leq \frac{1}{\beta} \pi(z^*)
\end{align*}

\subsection{Combining the two algorithms}
\vspace{-2mm}
To combine the LMP approximation with threshold rounding we require a slight modification of the instance
submitted to the LMP approximation.

Recall that for an instance $I$ the LMP 3-approximation returns a solution $T$ such that
$w(T) + 3 \pi(V\setminus T) \leq 3 OPT_a$.
Consider now instance $I'$ with has its penalties scaled by $1/3$, \ie, $\pi'_v = \frac{1}{3} \pi_v$.
Run the LMP approximation on $I'$ to obtained solution $T'$ satisfying
$w(T') + \pi(V\setminus T') = w(T') + 3 \pi'(V\setminus T') \leq 3OPT'_a$,
where $OPT'_a$ is the value of the optimum solution to program $LP'_a$ derived from $LP_a$ by taking scaled penalties $\pi'$.
Observe that $(x^*, z^*)$ is also feasible to $LP'_a$, because this program differs only in the objective function. Hence we have that
\begin{align*}
	w(T') + \pi(V\setminus T') & \leq 3OPT'_a \leq 3\left(w(x^*) + \pi'(z^*)\right) = 3 w(x^*) + \pi(z^*)
\end{align*}
Now, our final algorithm returns the best solution among $T'$ and the solution produced by the threshold rounding technique
in the previous section. Note that this is a deterministic procedure. However, the analysis uses a randomized argument inspired
by Goemans~\cite{Combining}: pick one solution with probability $p$ and the other with probability $1-p$.
Let $SOL$ be the returned solution.
\begin{align*}
	\mathbb{E}\left[SOL\right]
	& \leq
	\left[3p + (1-p) \frac{2.4}{\beta} \ln\left(\frac{1}{1-\beta}\right) (1+\epsilon)^2 \right] w(x^*) + \left[p + (1-p)\frac{1}{\beta}\right]\pi(z^*) \\
	& \leq
	(1+\epsilon)^2 \left[\left(3p{+}(1{-}p) \frac{2.4}{\beta} \ln\left(\frac{1}{1{-}\beta}\right) \right) w(x^*){+}\left(p{+}(1{-}p)\frac{1}{\beta}\right) \pi(z^*)\right]
\end{align*}
Finally, optimizing constants we obtain for $\beta = 1 - e^{-\frac{5}{36}}$ and $p = \frac{1}{4-3e^{-5/36}}$ the claimed result
\begin{align*}
	\mathbb{E}\left[SOL\right]
	& \leq
	\frac{4}{4-3e^{-5/36}} (1+\epsilon)^2 \left(w(x^*) + \pi(z^*)\right) \\
	& \leq
	\frac{4}{4-3e^{-5/36}} (1+\epsilon)^2 OPT \approx (2.8797 + \epsilon') \cdot OPT
\end{align*}

\appendix

\section{Adapted proof from~\cite{Moldenhauer}} \label{App:AppendixB}

We outline the proof of Lemma~\ref{lemMainCC}. As indicated this proof is, except for a minor change, analog
to the proof used in~\cite{Moldenhauer} to show that the generic primal-dual algorithm for Node-weighted Steiner Forest
on planar graphs has an approximation guarantee of~3.

Let $F'$ be the output of the algorithm and $A_i$ be the set of active moats before running iteration $i$.
We want to show that
\begin{align}
	\sum_{S\in A_i \cap CC} | F' \cap \Gamma(S) | \leq 3 | A_i \cap CC |. \label{cceq}
\end{align}

Within the rest of the proof, restrict to the induced subgraph of the union of all moats in $CC$ and $r$, \ie, discard
from $G$ every vertex that is not the root and not in any moat of $CC$.
Then, in~\eqref{cceq} we count the adjacencies between active moats at iteration $i$ and vertices from $F'$.
Let $F_i$ be the set of vertices bought by the algorithm before iteration $i$.
Consider a graph $G'$ obtained from $G$ in the following way:
\begin{enumerate}
	\item take the subgraph of $G$ induced by vertices from $F_i \cup F'$
	\item contract each inactive moat (at iteration $i$) in this subgraph with a neighboring vertex (excluding the moat containing root)
	\item contract each active moat in this component
	\item contract the moat containing the root
\end{enumerate}
Next, color the vertices of $G'$ with three colors:
\begin{itemize}
	\item white color for vertices obtained from contracting active moats
	\item blue color for the single vertex representing the moat containing the root
	\item black color for all other vertices, i.e. $F' \setminus F_i$
\end{itemize}
Observe now that deleting a black vertex in $G'$ disconnects some white vertex from the blue vertex,
because otherwise it would be deleted in the pruning phase.
$G'$ remains planar, since deletions and contractions preserve planarity. Moreover, it is easy to see
that the number of adjacencies $\sum_{S\in{A_i}}|F' \cap \Gamma(S)|$ in $G$ is the same as the number of
edges between white and black vertices in $G'$.

To bound this number we will use the following result that is implicit in~\cite{Moldenhauer}.
\begin{lemma}
	\label{lemma-graph-white-black}
	Consider a simple connected planar graph $H=(V,E)$ in which vertices are colored with two colors: black and white, i.e. $V=B\cup W$. If for this graph the two following conditions hold
	\begin{itemize}
		\item{there is no edge between any two white vertices}
		\item{removing any black vertex disconnects the graph}
	\end{itemize}
	then the number of edges between black and white vertices ($|E'|$) is at most 3 times greater then the number of white vertices, i.e., $|E'| \leq 3 (|W|-1)$
\end{lemma}
Before we prove the lemma, let us remark how it yields the claim.
Consider for a moment the color of the blue vertex in $G'$ to be white (resulting in graph $H$).
Now removing a black vertex clearly splits the graph into multiple components, since it disconnects
at least two white vertices (one of them is this recolored blue vertex). All other conditions of the lemma are satisfied.
Applying Lemma~\ref{lemma-graph-white-black} finishes the proof, since $|A_i| = |W|-1$.

\begin{proof} [Proof of the Lemma~\ref{lemma-graph-white-black}]
	We follow the proof given in \cite{MoldenhauerThesis}.\\
	Consider the following operation on the graph $H$.\\
	Take any edge $e=(u,v)$ between two black vertices $u$ and $v$ in $H$.
	\begin{itemize}
		\item{If $u$ and $v$ share a common white neighbor, then delete edge $e$.}
		\item{Otherwise contract $u$ and $v$.}
	\end{itemize}
	Observe that this operation preserves conditions of the lemma. Moreover it does not change the number
	of adjacencies between black and white vertices. Consider now the graph $H'$ obtained by performing as many
	above operations as possible. The $H'$ is bipartite since we contracted or deleted all edges between any
	two black vertices. The goal is now to bound the number of edges in $H'$.
	The idea is to use the Euler's formula for planar graphs. But first we have to show a few claims about $H'$.

	Let $W$ and $B$ denote the set of white and black vertices of $H'$, respectively.
	\begin{fact}
		$|B| \leq |W|-1.$ \label{substitutefact}
	\end{fact}
	\begin{proof}
		Consider a bread-first search tree $T$ in $H'$ rooted at any white vertex $r_w$. Since removing a black vertex splits the graph, all leaves of $T$ are white. Recall that $H'$ is bipartite. Thus each black vertex has at least one unique white child in $T$. Furthermore, $r_w$ is the only white vertex that does not have a parent. This concludes the fact.
	\end{proof}
	Now, using Fact~\ref{substitutefact} instead of Claim~1.4 of~\cite{MoldenhauerThesis} in the proof of Lemma~1.3 of~\cite{MoldenhauerThesis}
	yields the result.
\end{proof}

\section{The primal-dual 4-approximation for forest} \label{App:AppendixForest}
	In this section we use a general combinatorial approach for solving prize-collecting problems introduced by Hajiaghayi and Jain \cite{Hajiaghayi}. In their work they obtained the primal-dual $3$-approximation algorithm for edge-weighted Prize-collecting Steiner Forest problem. We repeat their argumentation in planar node-weighted setting resulting in the $4$-approximation algorithm. 

Consider a graph $G=(V,E)$ with a non-negative cost function on nodes $w : V \rightarrow Q_+$, a set of pairs of vertices (demands) $D = {(s_1,t_1), (s_2,t_2),\dots, (s_k,t_k)}$ and a non-negative penalty function $\pi : D \rightarrow Q_+$. In the Node-weighted Prize-collecting Steiner Forest problem we are asked to find a set of vertices $F\subseteq V$ which minimizes the sum of costs of vertices in $F$ plus penalties for pairs of vertices which are not connected in a subgraph of $G$ induced by $F$.

Note that we can give an equivalent definition of demands and penalties by specifying penalties for each unordered pair of vertices. Simply set penalties for pairs of vertices which are not in $D$ to $0$. From now on we will use values $\pi_{ij}$ to denote penalties. Let also $\Gamma(S)$ denote the set of vertices in $V-S$ incident to vertices from $S \subseteq V$ and let $S \odot (i,j)$ means that $|(i,j)\cap S| = 1$ (i.e., $S$ separates vertices $i$ and $j$)
Using this notation, we can formulate our problem with the following integer program
\begin{align*}
\text{min} & \sum_{v\in V} w_v x_v + \sum_{(i,j)\in V\times V}\pi_{ij}z_{ij} & (IP_{SF})\\
s.t.\\
& \sum_{v\in \Gamma(S)} x_v + z_{i,j} \geq  1			& \forall S \subseteq V, \hspace{3mm} \forall (i,j) \in V \times V : \hspace{3mm} S \odot (i,j) \\
& x_v \in \{0,1\}										& \forall v \in V\\
& z_{i,j} \in \{0,1\}									& \forall (i,j) \in S \times S
\end{align*}
Setting $x_v = 1$ corresponds to buying a vertex $v$ (including $v$ into solution $F$) and setting $z_{i,j} = 1$ corresponds to paying a penalty instead of connecting vertices $i$ and $j$.

The dual of the linear relaxation of this program is:
\begin{align*}
\text{max} & \sum_{S \subseteq V, S\odot (i,j)} y_{S_{ij}} & (DLP_{SF} 1)\\
s.t.\\
& \sum_{S:v\in\Gamma(S), S\odot(i,j)} y_{S_{ij}} \leq w_v	& \forall v \in V\\
& \sum_{S:S\odot(i,j)} y_{S_{ij}} \leq \pi_{i,j}			& \forall (i,j) \in V \times V\\
& y_{S_{ij}} \geq 0											& \forall S \subseteq V, S \odot (i,j)\\
\end{align*}
The problem with this dual program is that it has many different variables for each pair of vertices. Hence in a moat growing approach we have to decide how to split the growth of a moat corresponding to a set $S$ between variables $y_{S_{ij}}$. It seems to be a difficult task (see \cite{Hajiaghayi} for a detailed discussion) and may require decreasing some dual variables throughout the course of the algorithm.

Fortunately Hajiaghayi and Jain in \cite{Hajiaghayi} proposed a general approach of handling this issue of different variables induced by prize-collecting setting by circumventing it using Farkas' Lemma. The following arguments are repetitions of their work in the node-weighted setting and we conduct them for the sake of the completeness.

First, we create new variables $y_S = \sum_{(i,j): S\odot(i,j)} y_{S_{i,j}}$. Now the dual $DLP_{SF} 1$ becomes
\begin{align*}
\text{max} & \sum_{S \subseteq V} y_S & (DLP_{SF} 2)\\
s.t.\\
& y_S \leq \sum_{(i,j): S\odot(i,j)} y_{S_{i,j}}			& \forall S \subseteq V\\
& \sum_{S:v\in\Gamma(S)} y_S \leq w_v						& \forall v \in V\\
& \sum_{S:S\odot(i,j)} y_{S_{ij}} \leq \pi_{i,j}			& \forall (i,j) \in V \times V\\
& y_{S_{ij}} \geq 0											& \forall S \subseteq V, S \odot (i,j)\\
& y_S \geq 0												& \forall S \subseteq V\\
\end{align*}
\begin{fact}
Linear programs $DLP_{SF} 1$ and $DLP_{SF} 2$ are equivalent
\end{fact}
\begin{proof}
Take a feasible solution $y$ to $DLP_{SF} 1$. Let $y_S = \sum_{(i,j): S\odot(i,j)} y_{S_{i,j}}$. This together with $y$ constitutes a feasible solution to $DLP_{SF} 2$ of the same cost.\\
To see the other direction, take a feasible solution $y$ to $DLP_{SF} 2$. We can assume that the first constraint in $DLP_{SF} 2$ is tight, because otherwise we could decrease $y_{S{i,j}}$ until $y_S = \sum_{(i,j): S\odot(i,j)} y_{S_{i,j}}$ without affecting the objective function and not violating other constraints. The $y_{S{i,j}}$ is feasible to $DLP_{SF} 1$ and the objective function is the same. \qed
\end{proof}

Now we will use Farkas' Lemma to get a rid of dual variables $y_{S_{i,j}}$. Observe that they are not included in the objective function of $DLP_{SF} 2$. The idea is to replace constraints involving $y_{S_{i,j}}$ with different inequalities which for fixed $y_S$ check whether feasible $y_{S_{i,j}}$ exists.

\begin{fact} Farkas' lemma (variant)\\
	\label{FarkasLemma}
	Consider a matrix $A \in R^{m \times n}$ and a vector $b \in R^m$. The system $Ax\leq b$ has a solution $x \geq 0$, if and only if for all $y\geq 0$ with $yA \geq 0$ one has $yb \geq 0$.
\end{fact}

Consider a feasible solution to $DLP_{SF} 2$ and a system defined by constraints of $DLP_{SF} 2$ containing $y_{S_{i,j}}$, i.e.
\begin{align*}
& \sum_{(i,j): S\odot(i,j)} y_{S_{i,j}} \geq y_S			& \forall S \subseteq V\\
& -\sum_{S:S\odot(i,j)} y_{S_{ij}} \geq -\pi_{i,j}			& \forall (i,j) \in V \times V\\
\end{align*}
Farkas Lemma (see Fact~\ref{FarkasLemma}) says that this system has a solution $y_{S_{ij}} \geq 0$ if and only if for each vector $[\alpha | \beta] \geq 0$ with $\alpha_S - \beta_{i,j} \geq 0$ (for each $S, i,j$ such that $S\odot(i,j)$) we have that $\sum_{S\subseteq V} \alpha_S - \sum_{i,j} \beta_{i,j} \pi_{i,j} \leq 0$.
Notice that we can safely replace $\beta_{i,j}$ with $\max_{S: S\odot(i,j)}{\alpha_S}$ which gives us the following constraint:
\begin{align*}
& \sum_{S\subseteq V} \alpha_S \cdot y_S \leq \sum_{i,j} \max_{S: S\odot(i,j)}{\alpha_S} \cdot \pi_{i,j} & \text{for each} \hspace{3mm}\alpha: 2^V \rightarrow \mathbb{R}^+\\
\end{align*}
So our new dual is
\begin{align*}
\text{max} & \sum_{S \subseteq V} y_S & (DLP_{SF} 3)\\
s.t.\\
& \sum_{S:v\in\Gamma(S)} y_S \leq w_v						& \forall v \in V\\
& \sum_{S\subseteq V} \alpha_S \cdot y_S \leq \sum_{i,j} \max_{S: S\odot(i,j)}{\alpha_S} \cdot \pi_{i,j} & \text{for each} \hspace{3mm}\alpha: 2^V \rightarrow \mathbb{R}^+\\
& y_S \geq 0												& \forall S \subseteq V\\
\end{align*}
and it has only one dual variable $y_S$ for each set $S$. On the other hand, it has infinitely many constraints. However, as the lemma below says, many of them are redundant.
\begin{lemma} (Lemma 2.2 in \cite{Hajiaghayi})
It is sufficient to consider $\alpha's$ having only one positive value in its range.
\end{lemma}

The above lemma allows us to think about $\alpha$'s as families of subsets of $V$. Hence we can write our dual as follows
\begin{align}
\text{max} & \sum_{S \subset V} y_S & (DLP_{SF} 4) \nonumber\\
s.t. \nonumber\\
& \sum_{S:v\in\Gamma(S)} y_S \leq w_v																& \forall v \in V \label{ForestConstraint1}\\
& \sum_{S\in \mathbb{S}} y_S \leq \sum_{(i,j) \in V \times V, \mathbb{S}\odot(i,j)} \pi_{i,j}		& \forall \mathbb{S} \in 2^{2^V}\label{ForestConstraint2}\\
& y_S \geq 0																						& \forall S \subset V \nonumber
\end{align}
where $\mathbb{S}\odot(i,j)$ denotes that there exists $S \in \mathbb{S}$ such that $S \odot (i,j)$ (we say that family $\mathbb{S}$ separates vertices $i$ and $j$ if and only if there exists at least one set $S \in \mathbb{S}$ which separates vertices $i$ and $j$).

Note that $\mathbb{S}$ is a family of subsets of vertices and our dual has double exponential number of constraints. But we have now only one dual variable for each set. Intuitively this double exponential number of constraints implicitly ensures that for given variables $y_S$ there exist feasible variables $y_{S_{ij}}$ of the former dual program which sum to $y_S$.

Although the double exponential number of constraints does not sound good, we will be able to construct polynomial-time primal-dual algorithm based on this dual.

We can define a function $f : 2^{2^V} \rightarrow \mathbb{R}_+$ which for every family $\mathbb{S}$ define $f(\mathbb{S})$ to be the right-hand side of the corresponding constraint, i.e.:
$$f(\mathbb{S}) = \sum_{(i,j) \in V \times V, \mathbb{S}\odot(i,j)} \pi_{i,j}$$
In their paper, Hajiaghayi and Jain show that $f$ is submodular. This property allows them to prove the following fact.
\begin{fact} (Corollary 2.2 in \cite{Hajiaghayi})
\label{SumTight}
Suppose $y$ is a feasible solution to dual $DLP_{SF} 4$. Suppose the constraints corresponding to families $\mathbb{S}_1$ and $\mathbb{S}_2$ are tight. Then the constraint corresponding to the family $\mathbb{S}_1 \cup \mathbb{S}_2$ is also tight.
\end{fact}
\newpage
\subsection{Algorithm}
Without loss of generality we can assume that each terminal $v$ belongs to exactly one demand and its weight $w_v$ is $0$. To see this, construct a new graph where for each vertex $v_i, v_j$ of each demand $(i,j) \in D$ we have additional two vertices $v_i^{ij}$ and $v_j^{ij}$ connected by a single edge to original vertices ($v_i$ and $v_j$ correspondingly). The penalties are now only between vertices $v_i^{ij}$ and $v_j^{ij}$. Weights of new vertices are now $0$ while weights of original vertices $v_i, v_j$ remain the same. It is easy to see that every solution for the new graph can be used to construct a solution of the same cost for the original graph and vice-versa.

Now we are ready to give a primal-dual algorithm for the NWPCSF problem on planar graphs.
The algorithm starts with an initial solution $F$ in which there are all vertices of cost $0$ (hence all terminals). In each iteration the algorithm maintains moats which are the connected components of graph $G$ induced by the vertices of the current solution $F$. Demands can be marked (meaning that we decide to pay a penalty for them) or unmarked. At the beginning all demands are unmarked. Once demand is marked, it stays marked forever. A moat (denoted by the corresponding set $S\subseteq V$) is active in the current iteration if and only if there is at least one unmarked demand $(i,j)$ such that $S \odot (i,j)$. Now in each iteration we simultaneously grow each active moat until one of the following two events occur:
\begin{itemize}
	\item a vertex $v$ goes tight (constraint (\ref{ForestConstraint1}) becomes equality), or
	\item a family $\mathbb{S}$ goes tight (constraint (\ref{ForestConstraint2}) becomes equality).
\end{itemize}

In the first case we simply add $v$ to our solution $F$ (which may make some moats inactive) and continue to the next iteration.

In the second case, we mark each demand $(i,j)$ such that $\mathbb{S} \odot (i,j)$. Hence in the following iterations all moats from $\mathbb{S}$ will be inactive, and we will not violate any constraint during the growth process. We repeat this process until all moats become inactive.

After that we have an additional pruning phase in which we process all vertices of $F$ in the reverse order of buying. We remove a vertex $v$ from $F$ if after its removal from $F$, all unmarked demands are still connected in the graph induced by $F$. We output this pruned set of vertices as $F'$ which is our final solution.

Obtaining $\epsilon_1$ and a tight vertex in line $7$ is straightforward. On the other hand obtaining $\epsilon_2$ in line $8$ and a tight family $\mathbb{S}$ seems to be much harder, since the number of corresponding constraint is double exponential. Fortunately Hajiaghayi and Jain in section $4$ of \cite{Hajiaghayi} gave a polynomial time algorithm for computing $\epsilon_2$ and the corresponding tight family $\mathbb{S}$.

Since the algorithm terminates after at most $2|V|-1$ iterations (in each iteration the number of active moats or the number of connected components decreases), the running time of this algorithm is polynomial.

\begin{algorithm}[H]
	\SetKwData{Left}{left}\SetKwData{This}{this}\SetKwData{Up}{up}
	\SetKwInOut{Input}{input}\SetKwInOut{Output}{output}
	\Input{A planar graph $G=(V,E)$ with non-negative weights $w_i$ on the nodes and non-negative penalties $\pi_{ij}$ between each pair of vertices such that if $\pi_{ij} > 0$ then $w_i=0$ and $w_j=0$}
	\Output{A set of vertices $F'$ representing a forest and a set of pairs $Q'$ representing not connected demands}
	\BlankLine
	\Begin{
		$F\leftarrow \{v_i\in V: w_i = 0\}$\;
		$Q\leftarrow \emptyset$ \tcp{set all demands unmarked}
		$y_S\leftarrow 0$ \tcp{implicitly}
		\BlankLine
		\BlankLine
		$AM \leftarrow \left\{S\subseteq V : S \in SCC\left(G[F]\right) \wedge \displaystyle\mathop{\exists}_{(i,j) \in V\times V - Q } \pi_{ij}>0 \wedge S \odot (i,j) \right\}$\;
		\tcp{identify active moats as components of subgraph of $G$ induced by vertices $F$ for which there is at least one unmarked demand $(i,j)$ which is separated by the corresponding set}
		\BlankLine
		\BlankLine
		\While{AM $ \neq \emptyset$}{
			find minimum $\epsilon_1$ s.t if we increase $y_S$ for each $S\in AM$ by $\epsilon_1$ we get a new tight vertex $v$\;
			find minimum $\epsilon_2$ s.t if we increase $y_S$ for each $S\in AM$ by $\epsilon_2$ we get a new tight family $\mathbb{S}$\;
			$\epsilon \leftarrow min(\epsilon_1,\epsilon_2)$\;
			$y_S \leftarrow y_S + \epsilon$ for all $S\in AM$\;
			\eIf{$\epsilon = \epsilon_1$}{
				$F \leftarrow F \cup \{v\}$\;
			}{
				$Q \leftarrow Q \cup \{(i,j) \in V\times V : \mathbb{S} \odot (i,j)\}$
			}
			$AM \leftarrow \left\{S\subseteq V : S \in SCC\left(G[F]\right) \wedge \displaystyle\mathop{\exists}_{(i,j) \in V\times V - Q } \pi_{ij}>0 \wedge S \odot (i,j) \right\}$\;
		}
		\tcp{pruning phase}
		Derive $F'$ from $F$ by removing vertices in reverse order of purchase so that every unmarked demand is connected in $F'$.\\
		Let $Q'$ be all demands not connected via $F'$
	}
	\caption{Primal-Dual Algorithm for NWPCSF on planar graphs}
\end{algorithm}

\subsection{Analysis}
\begin{theorem}
\label{4aprox}
The algorithm outputs a set of vertices $F'$ and a set of demands $Q'$ which are not connected via $F'$ such that
$$\sum_{v\in F'}w_v + \sum_{(i,j)\in Q'}\pi_{ij} \leq 4\sum_{S\subseteq V}y_S \leq 4\ OPT$$
\end{theorem}
In order to prove Theorem~\ref{4aprox} it is enough to prove the following two lemmas:
\begin{lemma}
\label{pc1lemma}
$$\sum_{(i,j)\in Q'}\pi_{ij} \leq \sum_{S\subseteq V}y_S$$
\end{lemma}
\begin{proof}
First observe that $Q' \subseteq Q$ where $Q$ are marked pairs. Now consider families $\mathbb{S}_1,\dots,\mathbb{S}_f$ which went tight during the run of the algorithm. Observe that each marked pair was separated by some $\mathbb{S}_i$.\\
Hence family $\mathbb{S}_{all} = \displaystyle\mathop{\cup}_{j=1}^{f}\mathbb{S}_j$ separates each marked pair.
From Fact~\ref{SumTight} the union of tight families is tight. Putting it all together gives:\\
$$\sum_{(i,j)\in Q'}\pi_{ij} \leq \sum_{(i,j)\in Q}\pi_{ij} \leq \sum_{\mathbb{S}_{all}\odot (i,j)}\pi_{ij} = \sum_{S\in\mathbb{S}}y_S \leq \sum_{S\subseteq V}y_S$$ \qed
\end{proof}

\begin{lemma}
\label{general3lemma}
$$\sum_{v\in F'}w_v \leq 3\sum_{S\subseteq V}y_S$$
\end{lemma}
To prove Lemma~\ref{general3lemma} we will use an auxiliary lemma. But first, let us introduce one definition.\\
For a set of nodes $F$ and the set of unmarked demands $R=D-Q$ define a minimal feasible augmentation $F_{aug}$ of $F$ with respect to $R$ to be a set of vertices $F_{aug}$ containing $F$ as a subset such that every pair of vertices from $R$ is connected in the subgraph of $G$ induced by $F_{aug}$ and such that removal of any $v\in F_{aug} - F$ from $F_{aug}$ disconnects some pair from $R$.

\begin{lemma}
\label{Carsten3lemma}
Let $G$ be planar, $R$ be the set of unmarked demands after running the above algorithm, $F_j$ be the set of bought vertices before running iteration $j$ and $F_{aug}$ be a minimal feasible augmentation of $F_j$ with respect to $R$. Let also $A_j$ be the set of active moats before running iteration $j$. Then\\
$$\sum_{S\in{A_j}}|F_{aug} \cap \Gamma(S)| \leq 3|A_j|$$
\end{lemma}

Before we prove Lemma~\ref{Carsten3lemma} we will show how it helps us in proving Lemma~\ref{general3lemma}.\\
\begin{proof}[Proof of Lemma~\ref{general3lemma}]\\
Since we add a vertex $v$ to $F$ only if it is tight, and after that we do not modify variables corresponding to sets adjacent to $v$, we have the following equality
\begin{align*}
\sum_{v\in F'}w_v & = \sum_{v\in F'} \sum_{S:v\in\Gamma(S)}y_S = \sum_{S\subseteq V}|F' \cap \Gamma(S)|y_S
\end{align*}
(in the last step we changed the order of the summation).

Now let $F'$ be the output of the algorithm, $R=D-Q$ be the set of all unmarked demands, $F_j$ be the set of bought vertices before iteration $j$, $A_j$ be the set of active moats before running iteration $j$ and $\epsilon_j$ be the increase of the dual variables in iteration $j$. Then for each $S\subseteq V$ we have $y_S = \sum_{j:S \in A_j} \epsilon_j$ hence the following holds:
$$\sum_{S\subseteq V} y_S = \sum_j |A_j|\epsilon_j$$
and
\begin{align*}
\sum_{v\in F'}w_v & = \sum_{S\subseteq V}|F' \cap \Gamma(S)|y_S  = \sum_{S\subseteq V}|F' \cap \Gamma(S)|\sum_{j : S\in A_j}\epsilon_j\\
	& = \sum_j\left(\sum_{S\in A_j}|F' \cap \Gamma(S)|\right)\epsilon_j
\end{align*}

Observe now, that $F_{aug} = F_j \cup F'$ is a minimal feasible augmentation of $F_j$ with respect to $R$. Obviously, every demand from $R$ is connected in $F_{aug}$. Consider any vertex $v \in F' - F_j$. Removing $v$ from $F_{aug}$ will make some pair from $R$ disconnected because otherwise $v$ would be deleted in the pruning phase. Hence we can use Lemma~\ref{Carsten3lemma}:
$$\sum_{S\in{A_j}}|(F_j \cup F') \cap \Gamma(S)| \leq 3|A_j|$$
Since $|F' \cap \Gamma(S)| \leq |(F_j \cup F') \cap \Gamma(S)|$ we have
\begin{align*}
\sum_{v\in F'}w_v & = \sum_j\left(\sum_{S\in A_j}|F' \cap \Gamma(S)|\right)\epsilon_j\\
	& \leq \sum_j3|A_j|\epsilon_j = 3\sum_{S\subseteq V}y_S
\end{align*} \qed
\end{proof}

Now we conclude with the sketch of the proof of Lemma~\ref{Carsten3lemma}.
\begin{proof}[of Lemma~\ref{Carsten3lemma}]
The proof is conducted in a similar way as the proof of Lemma~\ref{lemMainCC} and the analysis is essentially the same as in \cite{Moldenhauer}.
We need to count the adjacencies between active moats and vertices from $F_{aug} - F_j$. Consider the graph $G'$ obtained from $G$ in the following way:
\begin{enumerate}
	\item take the subgraph of $G$ induced by vertices from $F_{aug}$
	\item discard isolated inactive moats
	\item contract each inactive moat with a neighboring vertex 
	\item contract each active moat
\end{enumerate}
Next color vertices of $G'$ in two colors:
\begin{itemize}
	\item white color for vertices obtained from contracting active moats
	\item black color for all other vertices, i.e. $F_{aug} - F_j$
\end{itemize}
Observe now that deleting a black vertex in $G'$ disconnects two white vertices, because otherwise it would be deleted in the pruning phase. $G'$ remains planar, since deletions and contractions preserve planarity. Moreover, it is easy to see that the number of adjacencies $\sum_{S\in{A_j}}|(F_{aug}) \cap \Gamma(S)|$ in $G$ is the same as the number of edges between white and black vertices in $G'$. To bound this number we use Lemma~\ref{lemma-graph-white-black} for each component of $G'$. Therefore we have $$\sum_{S\in{A_j}}|(F_{aug}) \cap \Gamma(S)| \leq 3|A_j|$$ \qed
\end{proof}

\end{document}